\documentclass[a4paper, 10pt]{article}
\usepackage{authblk}

\usepackage{algorithm}
\usepackage{algorithmic}

\usepackage{stmaryrd}
\usepackage[british]{babel}
\usepackage[utf8]{inputenc}

\usepackage{amsmath,amsthm,amssymb,amstext, amscd}
\usepackage{tikz-cd}
\usepackage[colorlinks=true]{hyperref}
\usepackage{stmaryrd}
\usepackage{booktabs}
\usepackage{fouriernc}
\usepackage[T1]{fontenc}
\usepackage{nicefrac}

\usepackage{todonotes}

\theoremstyle{plain}
\newtheorem{Theorem}{Theorem}

\newtheorem{Lemma}[Theorem]{Lemma}

\theoremstyle{definition}
\newtheorem{Definition}[Theorem]{Definition}
\theoremstyle{remark}

\newcommand{\T}{\operatorname{T}}
\newcommand{\Co}{\operatorname{Co}}
\newcommand{\Rec}{\operatorname{Rec}}

\newcommand{\maxdatatype}{\D\times\D\times\N\times(\ID \to \ID)\times\Z^*\times(\Q\to\Q)}

\usepackage{listings}
\title{Implementing evaluation strategies for continuous real functions\textsuperscript{\small1}}
\author{Michal Kone\v{c}n\'y, Eike Neumann\\[1ex] \texttt{\small m.konecny@aston.ac.uk, neumaef1@gmail.com}}
\affil{Aston University, Birmingham, UK}
\makeindex
\date{}

\newcommand{\ie}{\textit{i.e.}, ~}
\newcommand{\eg}{\textit{e.g.}, ~}

\newcommand{\R}{\mathbb{R}}
\newcommand{\Z}{\mathbb{Z}}
\newcommand{\Q}{\mathbb{Q}}
\newcommand{\D}{\mathbb{D}}
\newcommand{\N}{\mathbb{N}}

\newcommand{\I}{{\mathcal{I}}}
\newcommand{\IQ}{{\I\Q}}
\newcommand{\ID}{{\I\D}}
\newcommand{\CI}{C\left([-1,1]\right)}

\newcommand{\dom}{\operatorname{dom}}

\newcommand{\var}{\operatorname{var}}
\newcommand{\eval}{\operatorname{eval}}

\renewcommand{\div}{\operatorname{div}}

\newcommand{\primit}{\operatorname{primit}}
\newcommand{\paramax}{\operatorname{paramax}}
\newcommand{\rangemax}{\operatorname{rangemax}}

\newcommand{\Fun}{\operatorname{Fun}}
\newcommand{\BFun}{\operatorname{BFun}}
\newcommand{\DBFun}{\operatorname{DBFun}}
\newcommand{\Poly}{\operatorname{Poly}}
\newcommand{\PPoly}{\operatorname{PPoly}}
\newcommand{\Frac}{\operatorname{Frac}}

\newcommand{\LPoly}{\operatorname{LPoly}}
\newcommand{\LPPoly}{\operatorname{LPPoly}}
\newcommand{\LFrac}{\operatorname{LFrac}}

\newcommand{\Set}[2]{\left\{ #1 \; \mid \; #2 \right\}}

\newcommand{\fr}{{\mathit{fr}}}
\newcommand{\FR}{{\mathit{FR}}}
\newcommand{\FD}{{\mathbb{FD}}}

\newcommand{\pp}{{\mathit{pp}}}
\newcommand{\PP}{{\mathit{PP}}}

\newcommand{\Local}{\operatorname{Local}}

\begin{document}

\maketitle

\begin{abstract}

We give a technical overview of our exact-real implementation
of various representations of the space of continuous
unary real functions over the unit domain and a family
of associated (partial) operations, including integration, 
range computation, as well as pointwise addition, multiplication, division, sine, cosine, square root and maximisation.

We use several representations close to the usual theoretical
model, based on an oracle that evaluates the function at a point or over an interval. 
We also include several representations based on an oracle that
computes a converging sequence of rigorous (piecewise or one-piece) polynomial and rational approximations over the whole unit domain.
Finally, we describe ``local'' representations that combine both approaches, ie oracle-like representations that return a rigorous symbolic approximation of the function over a requested interval sub-domain with a requested effort. 

See also our paper ``Representations and evaluation strategies for feasibly approximable functions'' which compares the efficiency of these representations and algorithms and also formally describes and analyses one of the key algorithms, namely a polynomial-time division of functions in a piecewise-polynomial representation.  We do not reproduce this division algorithm here.

\end{abstract}

\footnotetext[1]{\includegraphics[scale=0.04]{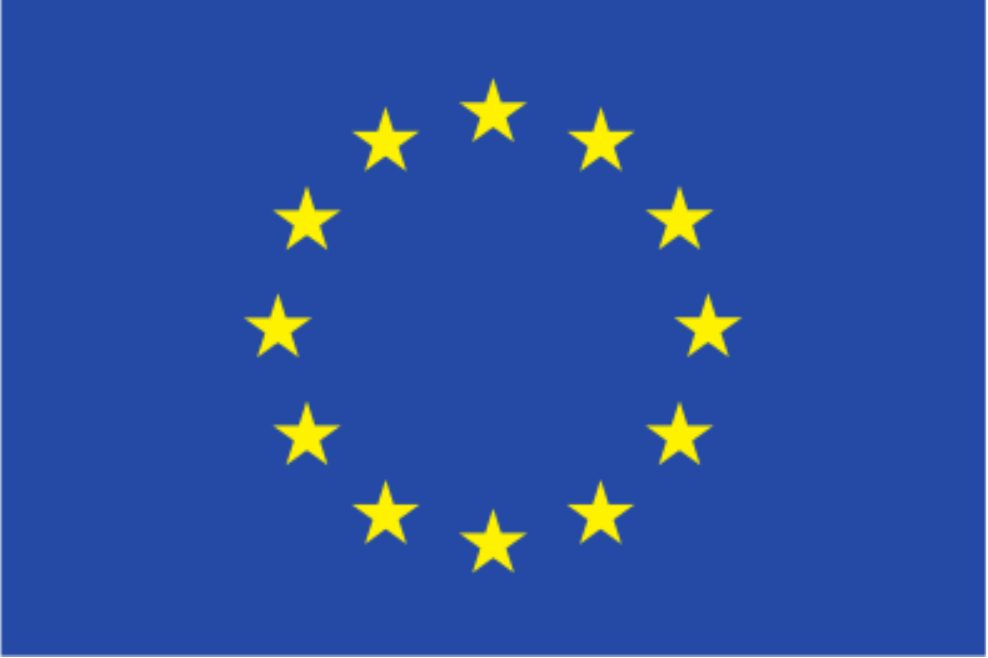} This project has received funding from the European Union’s Horizon 2020 research and innovation programme under the Marie Skłodowska-Curie grant agreement No 731143.}

\section{Exact real computation}

Exact real computation is an approach to numerical computation where real numbers appear as first class objects and the semantics of operations over them exactly agrees with the usual mathematical operations.
In practice, real number objects support an extraction of approximations to any requested arbitrarily high accuracy.
This approach is applied not only to real numbers but also to elements of other continuum spaces, for example, real continuous functions, real differentiable functions, real analytic functions and various subsets of Euclidean spaces.
Exact real computation is an implementation of the theory of computable analysis.
For background on computability in analysis see \eg \cite{SchroederPhD, PaulyRepresented, Weih, PourElRichards}.

We start with a brief overview of how real numbers are represented and used in  AERN2, our Haskell library for exact real computation.

\subsection{Balls}

An approximation to a real number with a known bound of the approximation error is represented by a ``ball'' \((c,e)\) 
with a centre $c\in\D$ and a non-negative radius $e\in \D$, where
$\D$ is the set of dyadic numbers.
Let $\ID$ be the set of such balls (equivalently, intervals with dyadic endpoints) and $\ID[-1,1] = \ID\cap [-1,1]$.

The centre and radius are represented in a floating-point-like format\footnote{Currently, the centre and radius are MPFR numbers by default.}.
The centre has an unlimited precision (except for physical computer constraints) and the radius has a mantissa precision fixed at 53 bits.
To emphasise the use of multi-precision numbers, the ball type is called
\lstinline{MPBall} in our Haskell implementation%
\footnote{To reproduce the script, start \texttt{ghci} using
\texttt{stack ghci aern2-fnreps:exe:aern2-fnreps-ops} or equivalent.}:

\newcommand{\ghciprompt}{\text{\tiny (ghci prompt)}}

\begin{lstlisting}{language=shell}
$\ghciprompt$> :t mpBall 1
mpBall 1 :: MPBall
\end{lstlisting}

An \lstinline{MPBall} implicitly holds a floating-point precision, which defines how much rounding is used in operations involving the ball.
In binary operations, the higher of the two precisions is used.
This is apparent in the following computations:

\begin{lstlisting}{language=shell}
$\ghciprompt$> t1 = mpBallP (prec 10) (1 /! 3)
$\ghciprompt$> t2 = mpBallP (prec 20) (1 /! 3)
$\ghciprompt$> t1
[0.3335 $\pm$ <2^(-11)]
$\ghciprompt$> t2
[0.33333349 $\pm$ <2^(-21)]
$\ghciprompt$> t1 + t2
[0.66683006 $\pm$ <2^(-10)]
$\ghciprompt$> getPrecision (t1 + t2)
Precision 20
\end{lstlisting}

\subsection{Bottom-up type derivations}

\lstinline{MPBall} implements the standard Haskell numerical type classes \lstinline{Num}, \lstinline{Fractional} and \lstinline{Floating}.  
Nevertheless, by default AERN2 modules use \texttt{mixed-types-num}\footnote{\url{https://hackage.haskell.org/package/mixed-types-num}}, an altered version of the Haskell standard Prelude, in which the numerical and related type classes are replaced by different ones, inspired by dynamically-typed languages.
The type of an expression is derived  bottom-up, \ie integer literals are always of type \lstinline{Integer} and rational literals of type \lstinline{Rational}, binary operations and relations support operands of different types, defining the result type depending on the operand types:

\begin{lstlisting}{language=shell}
$\ghciprompt$>  n = 1; q = 0.1
$\ghciprompt$> :t (n,q)
(n,q) :: (Integer, Rational)
$\ghciprompt$> :t n + q
n + q :: Rational
$\ghciprompt$> :t t1 + n
t1 + n :: MPBall
\end{lstlisting}

Partial operations such as division return values in an error-collecting monad, instead of throwing an exception:

\begin{lstlisting}{language=shell}
$\ghciprompt$> 1/0 :: CN Rational
{[(ERROR,division by 0)]}
$\ghciprompt$> 1/(t1-t1) :: CN MPBall
{[(POTENTIAL ERROR,division by 0)]}
\end{lstlisting}

The latter example above indicates a potential error instead of a certain error because the ball \lstinline{t1-t1} contains zero but also non-zero values and the ball \lstinline{t1} is treated as a set of values, in which only one is the intended one but we do not know which one.
The dependency between the two copies of \lstinline{t1} is lost.

The name \lstinline{CN} refers to a type function whose canonical name is much longer, which means that types reported by ghc are harder to read:

\begin{lstlisting}{language=shell}
$\ghciprompt$> :t 1/3
1/3 :: CollectErrors [(ErrorCertaintyLevel, NumError)] Rational
\end{lstlisting}

For convenient propagation, values in this monad can be used as operands:

\begin{lstlisting}{language=shell}
$\ghciprompt$> (1 + 1/3)/2 :: CN Rational
2 % 3
\end{lstlisting}

When there is no risk of a numerical error, one can use alternative,  exception-throwing operations:

\begin{lstlisting}{language=shell}
$\ghciprompt$> :t 1/!3
1/!3 :: Rational
\end{lstlisting}

\subsection{Partial ball comparisons}

Interval and ball comparisons can be blurred when balls overlap.
To facilitate safe and convenient ball comparisons, comparison relations over \lstinline{MPBall} do not return a \lstinline{Bool} as in standard Prelude but \lstinline{Maybe Bool}:

\begin{lstlisting}{language=shell}
$\ghciprompt$> 0 < t1
Just True
$\ghciprompt$> 1/3 < t1
Nothing
$\ghciprompt$> 0 == t1
Just False
$\ghciprompt$> t1 == t1
Nothing
\end{lstlisting}

For convenience, there are also ``surely'' and ``maybe'' comparison relations that return a \lstinline{Bool}:

\begin{lstlisting}{language=shell}
$\ghciprompt$> t1 !==! t1
False
$\ghciprompt$> t1 ?==? t1
True
\end{lstlisting}

\subsection{Real numbers}

Exact real numbers in AERN2 are represented as functions from an accuracy specification to an \lstinline{MPBall}.  
Accuracy is specified using \lstinline{bitsSG s g} where \lstinline{s} is a strict bound on the ball radius in terms of bits and \lstinline{g} is an indicative guide for the size of the error:

\begin{lstlisting}{language=shell}
$\ghciprompt$> :t pi
pi :: CauchyReal
$\ghciprompt$> pi ? (bitsSG 10 20) :: MPBall
[3.14159262180328369140625 $\pm$ <2^(-23)]
$\ghciprompt$> pi^!2 ? (bitsSG 10 20) :: MPBall
[9.869604110717769884786321199499070644378662109375 $\pm$ <2^(-19)]
\end{lstlisting}

Comparisons for real numbers are not decidable.  
This means that real number comparisons cannot have a simple Boolean return type.
In AERN2, such comparisons return a function from accuracy to \lstinline{Maybe Bool}, which allows us to try and decide comparison relation with a certain effort and test whether it succeeded or failed:

\begin{lstlisting}{language=shell}
$\ghciprompt$> (pi < pi + (0.1)^100) ? (bitsSG 10 20)
Nothing
$\ghciprompt$> (pi < pi + (0.1)^100) ? (bitsSG 1000 1000)
Just True
$\ghciprompt$> (pi < pi) ? (bitsSG 1000 1000)
Nothing
\end{lstlisting}

Also, due to the infinite nature of real numbers, partial functions into the real numbers usually cannot simply return a \lstinline{CN CauchyReal}.
Instead, partial real functions usually return a \lstinline{CauchyRealCN}, which encodes a function from accuracy to \lstinline{CN MPBall}:

\begin{lstlisting}{language=shell}
$\ghciprompt$> (sqrt pi) :: CauchyRealCN
[1.77245385090551602729816... $\pm$ <2^(-122)]
$\ghciprompt$> (sqrt pi ? (bitsS 10)) :: CN MPBall
[1.77245385083369910717010498046875 $\pm$ <2^(-32)]
\end{lstlisting}

The above example also demonstrates that the formatting function \lstinline{show}, when applied to a real number, uses a default accuracy.


\subsection{Landscape of representations}

Various approaches to representing real numbers have been proposed and implemented.
Some theoretically oriented works represent real numbers  as steams of signed binary digits or generalised digits.
For example, Escard\'o's RealPCF \cite{Escardo97-thesis} uses streams of contractive affine refinements and Potts et al's IC-Reals \cite{Potts98-thesis} uses streams of contractive linear fractional transformations.

According to the results of the competitions \cite{blanck2000era-competition, niqui2009many},  stream-based approaches tend to be relatively slow.  
The fastest implementations of exact real arithmetic seem to be iRRAM \cite{muller2000irram}, written in C++.
This package first uses interval/ball arithmetic at a fixed precision.
If the precision is insufficient to decide branch conditions or final results are not sufficiently accurate, the computation is scrapped and restarted with a higher precision.
This approach is repeated until the program succeeds.
In this approach there are no real number objects but the program as a whole has a real number semantics.
AERN2 allows one to write a program featuring real numbers as an abstract type and execute the same program either in the iRRAM manner or using exact real objects.

There seems to be a consensus that the most practically feasible representation of real numbers is via sequences or nets of interval/ball approximations, whether or not the sequences appear directly as first-class objects or is obtained indirectly via iRRAM-style re-computations with increasing precision.

The situation is much less clear regarding the representation of continuous real functions.
There are many candidate representations and comparing their computational complexity and practical performance is a matter of ongoing research.
This paper gives an overview of the most prominent candidate representations and their implementations in AERN2, while our other paper \cite{MainArticle} focuses on comparing these representations.
We first review which operations involving continuous functions are to be implemented for these representations.

\section{Operations of interest on $\CI$}

In this section we switch attention to the space $\CI$ of continuous real functions over the interval $[-1,1]$ and various operations of interest over this space.
First of all, it should be possible to evaluate a real function at a point:

\begin{itemize}
\item $\eval \colon \CI \times [-1,1] \to \R, \; f,x \mapsto f(x)$
\end{itemize}

Currently, our main applications for the arithmetic over $\CI$ are range computation and integration for unary real functions:

\begin{itemize}
\item $\rangemax \colon \CI \to \R, \; f \mapsto \max_{x\in[-1,1]}f(x)$
\item $\int \colon \CI \to \R, \; f \mapsto \int_{-1}^{1}f(x)\,\mathrm{d}x$
\end{itemize}

In AERN2 the operation $\rangemax$ is expressed as follows:

\begin{lstlisting}[language=Haskell]
-- type:
maximumOverDom :: 
  CanMaximiseOverDom f d => 
  f -> d -> MaximumOverDomType f d
-- usage:
m = maximumOverDom f (dyadicInterval (-1,1))
\end{lstlisting}

The type class constraint \lstinline{CanMaximiseOverDom f d} declares that the specific representation type \lstinline{f} can use the function \lstinline{maximumOverDom} with domains of type \lstinline{d}.
The type function \lstinline{MaximumOverDomType} specifies how the resulting real number will be represented.
The domain is typically an \lstinline{Interval Dyadic} and the result is typically either \lstinline{MPBall} or \lstinline{CauchyReal}.
If the function is a ball-like function approximation similar to \lstinline{MPBall}, the result is an \lstinline{MPBall}, and,
if the function is exact, the result is a \lstinline{CauchyReal}.

The integral operation is represented in AERN2 analogously, using the function \lstinline{integrateOverDom}, type class \lstinline{CanIntegrateOverDom}, and type function\\ \lstinline{IntegralOverDomType}.

Typically, a unary continuous real function is given by a symbolic expression with one real variable, such as:

\begin{lstlisting}[language=Haskell]
bumpy x = sin(10*x) `max` cos(11*x)
\end{lstlisting}

Nevertheless, sometimes a function is given without a symbolic representation, for example, if it is a solution of a differential equation or it comes from another ``black box'' or external source.
Here we focus on computation that can be applied not only to functions that are given symbolically, but to any (unary, bounded-domain) continuous real function.
(In the following section we define a number of representations that can accommodate all $\CI$ functions.)
We will therefore not make any use of a symbolic representation even when we have one.

To build new functions from existing functions, we should be able to apply common real operations pointwise to continuous functions:

\begin{itemize}
\item $+\colon \CI \times \CI \to \CI, \; (f,g) \mapsto f + g.$
\item $\times\colon \CI \times \CI \to \CI, \; (f,g) \mapsto f \cdot g.$ 
\item $- \colon \CI \to \CI, \; f \mapsto -f$.
\item $\div \colon \subseteq \CI \times \CI \to \CI, \; (f, g) \mapsto f/g$, where 
\[\dom(\div) = \Set{(f,g) \in \CI \times \CI}{g(x) \geq 1 \text{ for all }x \in [-1,1]}.\]
\item $\max \colon \CI \times \CI \to \CI, \; (f,g) \mapsto \max(f,g).$
\item $\sqrt{|\cdot|} \colon \CI \to \CI, \; f \mapsto \sqrt{|f|}.$
\end{itemize}

We also require pointwise applications of common trigonometric functions.
In Haskell and AERN2, there is no need to define new syntax for these pointwise function operations.
We simply make our function types instances of the numeric type classes and use the standard operation syntax.
For example, the parameter \lstinline{x} in for our polymorphic function \lstinline{bumpy} can be a function as well as a number:

\begin{lstlisting}
bumpy pi :: CauchyReal
x_BF :: UnaryBallFun
bumpy x_BF :: UnaryBallFun
\end{lstlisting}

where \lstinline{x_BF} is the identity function \(\lambda x.x\) over real numbers and \lstinline{UnaryBallFun} is one of our representations of $\CI$ functions.
The result function of type \lstinline{UnaryBallFun} is built by the pointwise function operations $\times$, $+$, $\sin$, $\cos$, and $\max$ as well as by implicit coercions of integers into constant functions.

The identity function over some interval domain \lstinline{d} is typically defined as follows:

\begin{lstlisting}
x_BP = varFn d ()
\end{lstlisting}

Here \lstinline{()} is a dummy variable name.
The polymorphic function \lstinline{varFn} requires a variable name so that it can be used also for building multi-variate projection functions \(\lambda x_1\ldots x_n.x_i\).
The type of the variable name depends on the function type.
In our case the variable type is always the unit type \lstinline{()} with the unique dummy value \lstinline{()}.

We consider also a few non-pointwise operations, namely composition, primitive function and parametric maximisation, although AERN2 does not yet have implementations of these for all of our representations:

\begin{itemize}
\item $\circ \colon \subseteq \CI \times \CI \to \CI, \; (f,g) \mapsto f \circ g,$ where 
\[\dom(\circ) = \Set{(f,g) \in \CI \times \CI}{g([-1,1]) \subseteq [-1,1]}. \]
\item $\primit \colon \CI \to \CI, \; f \mapsto \lambda t.\int_{-1}^t f(s) \operatorname{ds}.$
\item $\paramax \colon \CI \to \CI, \; f \mapsto \lambda t.\max\Set{f(s)}{s \in [-1,t]}.$
\end{itemize}

Note that differentiation is not a computable operation on the subset of differentiable functions in $\CI$.

\section{Representations of $\CI$}\label{section: representations}

In this section we introduce a number of representations of $\CI$ and
their AERN2 Haskell implementations.
The Haskell implementations actually support also partial real functions using the \lstinline{CN} monad for the result values.
Nevertheless, we will focus only on total functions in this paper.

Some representations encode a real function as a program-level function that returns approximations of the functions at different points or over small intervals.
Other representations use convergent collections of polynomials or similar approximations, each close to the function over the whole of its domain.
Finally, locally approximating representations combine features of both of these approaches.

\subsection{Point-evaluating representations}

Package \texttt{aern2-fun} defines the following representations of $\CI$\footnote{Actually, these representations support $C(D)$ over any compact real interval $D$, not only $[-1,1]$.}:

\begin{itemize}
  \item BFun, Haskell type \lstinline{UnaryBallFun}, encodes an $f\in\CI$ by a Haskell function $\varphi$ of type \lstinline{CN MPBall -> CN MPBall}, with the properties:
  \[
      \begin{array}{l}
      \forall x\in I \in \ID[-1,1].\quad 
      \text{if }
      \varphi(I)\text{ is defined, then } f(x) \in \varphi(I)
      \\[1ex]
      (\forall i. x\in I_i \subseteq [-1,1])
      \text{ and }
      (\lim_{i\to\infty} |I_i| = 0)
      \implies
      \\\qquad 
      \varphi(I_i) \text{ is defined for all sufficiently large }i
      \\\qquad
      \text{and }\lim_{i\to\infty} |\varphi(I_i)| = 0
      \end{array}
  \]
  \item DBFun, Haskell type \lstinline{UnaryBallDFun}, encodes an $f\in\CI$ by a pair \(\varphi, \varphi'\) of \lstinline{UnaryBallFun}, with $\varphi$ representing $f$ exactly as in BFun, and $\varphi'$ representing the (partially defined) derivative $f'$ in the following sense:
  \[
      \begin{array}{l}

      \text{Whenever } \varphi'(I)\text{ is defined,} 
      \\
      \text{the absolute value of } \varphi'(I)
      \text{ is a Lipschitz constant of } f \text{ over }I.
      \\[2ex]

      \text{For all } x\in \D[-1,1] \text{ where } f'(x) \text{ is defined, it holds:}
      \\[1ex]
      \quad\forall I \in \ID[-1,1].\quad 
      \text{if }
      x \in I \text{ and }
      \varphi'(I)\text{ is defined, then } f'(x) \in \varphi'(I)
      \\[1ex]
      \quad(\forall i. x\in I_i \subseteq [-1,1])
      \text{ and }
      (\lim_{i\to\infty} |I_i| = 0)
      \implies
      \\\quad\qquad 
      \varphi'(I_i) \text{ is defined for all sufficiently large }i
      \\\quad\qquad
      \text{and }\lim_{i\to\infty} |\varphi'(I_i)| = 0
      

      \end{array}
  \]
  (For convenience, the pair \(\varphi, \varphi'\) is encoded as two elements of a list.)

  \item Fun, Haskell type \lstinline{UnaryModFun}, encodes an $f\in\CI$ by a pair of Haskell functions:
  \[
    \begin{array}{lcl}
    \varphi : \D[-1,1] \rightharpoonup \R & \text{Haskell type} & \texttt{\lstinline{Dyadic -> CauchyRealCN}}
    \\
    \omega : \ID[-1,1] \to \N \to \N & \text{Haskell type} & \texttt{\lstinline{MPBall -> Integer -> Integer}}
    \end{array}
  \]
  with the properties:
  \[
    \begin{array}{l}
    \forall x\in \D[-1,1].\, \varphi(x) = f(x)
    \\[1ex]
    \forall x,y\in I \in \ID[-1,1].\, \forall n\in\N .\quad
    |x-y|\leq 2^{-\omega(I,n)} \implies |f(x)-f(y)| \leq 2^{-n}
    \end{array}
  \]
\end{itemize}

Variants of Fun are typically used in studies that focus on computability.
The function $\omega$ is a (localised) modulus of continuity of the encoded function.

Variants of BFun are typically used in implementations of exact real arithmetic.  
BFun is usually implemented via an arbitrary-precision interval arithmetic.
As there are good implementations of interval arithmetic, BFun is a sensible practical choice.
DBFun is a simple adaptation of BFun for (almost everywhere) differentiable functions.

\subsubsection{Basic Operations}

Pointwise operations, constant function and identity function constructions are implemented for BFun and DBFun by the usual interval extensions of the operations and functions (and also their derivatives in case of DBFun).
For example, pointwise division for DBFun is the following mapping on pairs of partial ball functions:
\[
  \frac{(f,f')}{(g,g')}  =
  \left(
    \lambda x.\frac {f(x)}{g(x)}, \;
    \lambda x.\frac{g(x)\cdot f'(x)-f(x)\cdot g'(x)}{g(x)^2}
  \right)
\]

Pointwise operations for Fun encodings are more subtle than those for BFun due to the need to define a valid modulus of continuity for the result.
For example, the Fun reciprocal is the following mapping of function encodings:
\[
  \frac{1}{(\varphi, \omega)} =
  \left(
      \lambda x. \frac{1}{\varphi(x)}, \lambda D i .  \left(i-2*\left(\mathop{\mathrm{lognorm}}(\varphi(D)) - 1\right) + 10\right)
  \right)
\] 
where $\mathrm{lognorm}(\varphi(D))$ returns an integer $n$ with $-(2^n) \leq \varphi(D) \leq 2^n$, and $n$ is the smallest or the second-smallest integer with this property.

\subsection{Globally approximating representations}

The representations Fun, BFun and DBFun all have in common that they are local in the sense that they describe a function a point at a time, or a small neighbourhood at a time.
The following representations are global in the sense that they provide information about a function over its complete interval domain \(D\) at once.
The simplest such representation describes a function by an indexed family of ``polynomial balls'' that converge to the function.

\subsubsection{Polynomial balls}

A polynomial ball is a pair \((p,e)\) where the centre $p$ is a univariate polynomial with dyadic coefficients and $e$ is a non-negative dyadic radius.
Let $\D[x]$ denote the set of polynomials with dyadic coefficients
and $\ID[x]$ the set of all polynomial balls.
For $f\in C(D)$ we write \(f\in [p\pm e]\) iff 
\(\forall x\in D.\, f(x)\in[p(x)\pm e]\).
\\
For $b=(p,e)\in\ID[x]$, let $|b| = 2e$.

We use the AERN2 type \lstinline{ChPoly MPBall} to represent \(\ID[x]\).
\lstinline{ChPoly t} is a type of polynomials in the Chebyshev basis with coefficients of type \lstinline{t}\footnote{%
Note that a polynomial has dyadic coefficients in the Chebyshev basis iff it has dyadic coefficients in the monomial basis.}.
Technically, \lstinline{ChPoly MPBall} stands for polynomials with ball coefficients, but we use only a subset which coincides with polynomial balls.
To do this, we enforce the invariant that all coefficients except the coefficient of the constant term have zero radius, \ie they are dyadic numbers.
The radius of the constant term is the radius of the polynomial ball.
The following example illustrates the basics of working with type \lstinline{ChPoly MPBall}: 

\begin{lstlisting}{language=shell}
$\ghciprompt$> x = varFn (unaryIntervalDom, bits 10) () :: ChPoly MPBall
$\ghciprompt$> (x+1)^!2
[1 $\pm$ 0] + [2 $\pm$ 0]*x + [1 $\pm$ 0]*x^!2

$\ghciprompt$> import AERN2.Poly.Cheb
Demo AERN2.Poly.Cheb> reduceDegree 1 ((x+1)^!2)
[1.5 $\pm$ <2^(-1)] + [2 $\pm$ 0]*x
\end{lstlisting}

Note that when constructing a \lstinline{ChPoly MPBall}, it is necessary to specify an interval domain for the ball.
Since the Chebyshev basis works well only on the domain $[-1,1]$, the type \lstinline{ChPoly t} contains a specification of a dyadic interval domain $[a,b]$ which is transparently translated to the internal domain $[-1,1]$.
In the above example we used the domain $[-1,1]$.  
Repeating it over the domain $[0,2]$ gives a different degree reduction:

\begin{lstlisting}{language=shell}
$\ghciprompt$> x = varFn (dyadicInterval (0,2), bits 10) () :: ChPoly MPBall
$\ghciprompt$> (x+1)^!2
[1 $\pm$ 0] + [2 $\pm$ 0]*x + [1 $\pm$ 0]*x^!2

$\ghciprompt$> import AERN2.Poly.Cheb
Demo AERN2.Poly.Cheb> reduceDegree 1 ((x+1)^!2)
[0.5 $\pm$ <2^(-1)] + [4 $\pm$ 0]*x
\end{lstlisting}

The quality of the polynomial degree reduction is down to using the Chebyshev basis.
In the monomial basis, a degree reduction over $[-1,1]$ would replace the quadratic term with $[1\pm 1]$, leading to the polynomial \lstinline{[1 $\pm$ <2^0] + [2 $\pm$ 0]*x}.
Over the domain $[0,2]$ it would lead to \lstinline{[3 $\pm$ <2^1] + [2 $\pm$ 0]*x}.

Normally, in AERN2 code we do not directly reduce the degree of polynomials.
Instead, there is automatic ``sweeping'', ie dropping insignificant Chebyshev terms, in most operations.
Terms are sweeped as much as possible while respecting a given \emph{accuracy guide}.
The accuracy guide is an accuracy value embedded in each polynomial ball.
In the above examples, it is part of the parameter for \lstinline{varFn}, namely \lstinline{bits 10}, which roughly corresponds to a permitted accuracy loss of $2^{-10}$.
In our simple example, we see the effect of the automatic sweeping only when we make the accuracy guide extremely loose:

\begin{lstlisting}{language=shell}
$\ghciprompt$> x = varFn (unaryIntervalDom, bits 10) () :: ChPoly MPBall
$\ghciprompt$> (x+1)^!2
[1 $\pm$ 0] + [2 $\pm$ 0]*x + [1 $\pm$ 0]*x^!2

$\ghciprompt$> x = varFn (unaryIntervalDom, bits (-2)) () :: ChPoly MPBall
$\ghciprompt$> (x+1)^!2
[1.5 $\pm$ <2^(-1)] + [2 $\pm$ 0]*x

$\ghciprompt$> x = varFn (unaryIntervalDom, bits (-10)) () :: ChPoly MPBall
$\ghciprompt$> (x+1)^!2
[1.5 $\pm$ <2^(1)]
\end{lstlisting}

\subsubsection{Rational function balls}

A rational function ball $\fr$ over the domain $D$ is a pair of polynomial balls 
\((p,e)\), \((q,d)\) over $D$ and a dyadic $m>0$ such that $|q(x)|-d > m$ for each dyadic $x \in D$.
Note that $q$ is either strictly positive or strictly negative over $D$.

Let $\FD[x]$ denote the set of all rational function balls over the domain $D$.
\\
For $f\in C(D)$ we write \(f\in \fr\) iff 
\(\forall x\in[-1,1].\, f(x)\in[p(x)\pm e]/ [q(x) \pm d]\).
\\
The nominal width of $\fr$ is defined as
\[
    |\fr| = \frac{\frac{M\cdot d}{m'} + e}{m'-d}
\]
where $m'$, respectively $M$, is the lower bound, respectively upper bound, 
of $|q|$ over $D$, computed by the maximisation algorithm described in Section~\ref{Section: Range computation}
with accuracy $n=2$.

Within our implementation, rational function balls are represented by the Haskell type \lstinline{Frac MPBall}.
A value of the type \lstinline{Frac MPBall} is a record formed of two \lstinline{ChPoly MPBall} values, one for the 
denominator and one for the numerator, and an \lstinline{MPBall} which is a lower bound to the denominator on the domain.

\subsubsection{Piecewise polynomial balls}

A piecewise polynomial ball $\pp$ over the domain $[-1,1]$ comprises:
\begin{itemize}
  \item
    A dyadic partition $-1 = a_0 < a_1 < \cdots < a_n = 1$, 
  \item
    A family $p_1,\ldots,p_n\in\ID[x]$.
\end{itemize}
While each $p_i$ is defined over $[-1,1]$, it is only used
over the domain $[a_{i-1},a_i]$.
For example, the width of $\pp$ is 
\[
  |\pp | = \max\big\{|p_i(x)| \;\vrule\; 1 \leq i \leq n, \, a_{i-1} \leq x \leq a_{i}  \big\}.
\]

In our implementation, piecewise polynomials are given by the Haskell type \lstinline{PPoly MPBall}.
A value of \lstinline{PPoly MPBall} is a record formed of a value of type \lstinline{DyadicInterval},
representing the domain, and a list of pairs of type \lstinline{(DyadicInterval, ChPoly MPBall)} representing
the partition of the domain together with the polynomial approximations on each piece of the partition. 

\subsubsection{Cauchy representations}

Sequences of polynomials or rational functions are used in the following representations of $\CI$:

\begin{itemize}
  \item Poly encodes an $f\in\CI$ as a sequence $f_P:\N \to \ID[x]$, Haskell type
  \lstinline{Accuracy -> ChPoly MPBall}, with \(\forall n\in\N.\,f\in f_P(n) \)
  and \(\lim_{n\to\infty} |f_P(n)| = 0\).

  \item PPoly encodes an $f\in\CI$ as a sequence $f_\PP$ of piecewise polynomial balls converging to $f$.
  The sequence has the Haskell type \lstinline{Accuracy -> PPoly}, with \(\forall n\in\N.\,f\in f_\PP(n) \)
  and \(\lim_{n\to\infty} |f_\PP(n)| = 0\).

  \item Frac encodes an $f\in\CI$ as a sequence $f_\FR$ of rational function balls converging to $f$.
  The sequence has the Haskell type \lstinline{Accuracy -> Frac}, with \(\forall n\in\N.\,f\in f_\FR(n) \)
  and \(\lim_{n\to\infty} |f_\FR(n)| = 0\).

\end{itemize}

Note that these sequences are not necessarily fast converging, \ie the rate of convergence is not necessarily $2^{-n}$. 
Although our implementations aim for this fast rate, we have not found an efficient way to guarantee it.

\subsection{Locally approximating representations}

Now we define representations that combine features of both point-evaluating representations and globally-approximating representations
where a polynomial or rational approximation with arbitrarily high accuracy is available over any dyadic sub-interval of the domain $[-1,1]$.

The representation $\LPoly$ encodes $f\in \CI $ by
a dependent-type function $F$ that maps each $D\in\ID$ to a Poly-name of $f|_D$.
Its Haskell type is\\ \lstinline{DyadicInterval -> Accuracy -> PPoly}.

Representations $\LPPoly$ and $\LFrac$ are defined analogously.


\section{Range computation}\label{Section: Range computation}

\subsection{Range computation for evaluation-based representations}

The representations $\Fun$, $\BFun$, and $\DBFun$ use a simple maximisation algorithm based on subdivision.
All three representations encode a function $f \colon [l,r] \to \R$ via an interval inclusion.
Our maximisation algorithm takes as input an interval inclusion and returns as output the global maximum as a real number.

Our maximisation algorithm has to take into account that the types \lstinline{UnaryBallFun} and \lstinline{UnaryBallDFun} produce outputs of type \lstinline{CN MPBall}, 
which may represent either an interval value or an error.
In order to model this mathematically, let us introduce the space $\ID_{\bot}$ of dyadic rational intervals with a bottom element $\bot$ added.
This bottom element is intended to represent an undefined value.
We say that a point $x \in \ID_{\bot}$ is \emph{defined} if it is different from $\bot$.
The arithmetic operations on $\ID$ extend to $\ID_{\bot}$ by letting the result of the operation be $\bot$ if any of the operands is $\bot$.
We say that $x \in \ID_{\bot}$ is \emph{certainly smaller than} $y \in \ID_{\bot}$ if both $x$ and $y$ are defined and the right endpoint of $x$ is smaller than the left endpoint of $y$.  

A \emph{maximisation segment} for a function $f \colon [l,r] \to \R$ is a tuple $(a,b,v) \in \D^2 \times \ID_{\bot}$ with 
$[a,b] \subseteq [l,r]$ and $f([a,b]) \subseteq v$.
We define a total preorder on the set of all maximisation segments of $f$ as follows: If $(a,b,v)$ and $(a',b',v')$ are maximisation segments then 
$(a,b,v) \leq (a',b',v')$ if and only if $v'$ is undefined or $v$ is defined and the right endpoint of $v$ is smaller than the right endpoint of $v'$.

\begin{algorithm}[H]\label{Algorithm: maximisation by subdivision}
  \caption{Maximisation}
  \begin{algorithmic}
  \REQUIRE An interval inclusion $F$ of a function $f \colon [l,r] \to \R$.
  An accuracy requirement $n \in \N$.
  \ENSURE An interval $I$ of radius smaller than $2^{-n}$ which contains the global maximum of $f$ over the interval $[l,r]$.
  \STATE \textbf{Procedure:}
  \STATE - Let $c = F([l,r])$.
  \STATE - Create a maximisation segment $M = (l, r, c)$
  \STATE - Create a priority queue $\mathbf{Q} = \{M\}$.
  \LOOP
    \STATE - Remove the largest element $M = (a,b,v)$ from the queue $\mathbf{Q}$.
    \STATE - Let $c = c \cap v$.
    \IF {$c$ has radius smaller than $2^{-n}$}
      \RETURN $c$.
    \ELSE 
      \STATE Let $m = (a + b) / 2$.
      \STATE Let $v_0 = F([a,m])$.
      \STATE Let $v_1 = F([m,b])$.
      \STATE Create a maximisation segment $M_0 = (a,m,v_0)$.
      \STATE Create a maximisation segment $M_1 = (m,b,v_1)$.  
      \STATE If $v_0$ is not certainly smaller than $c$ then add $M_0$ to the queue $\mathbf{Q}$.
      \STATE If $v_1$ is not certainly smaller than $c$ then add $M_1$ to the queue $\mathbf{Q}$. 
    \ENDIF
  \ENDLOOP
  \end{algorithmic}
  \end{algorithm}

\subsection{Real root counting}\label{Section: Root counting}

Our maximisation algorithm for representations based on (local) polynomial or rational approximations is essentially an improved version of Algorithm \ref{Algorithm: maximisation by subdivision}.
It enhances Algorithm \ref{Algorithm: maximisation by subdivision} with a real root counting technique that allows us to locally identify critical points and regions of monotonicity of the polynomial approximation.

This real root counting technique goes back to Vincent \cite{Vincent} and Uspensky \cite{Uspensky} and is based on counting sign variations in the Bernstein basis. 
We follow here the presentation in \cite[Chapter 10]{BasuPollackRoy} (see also the bibliographical notes there).

Let $P \in \Z[x]$ be a polynomial of degree at most $d$ with integer coefficients. Let $(a,b)$ be some bounded open interval. Our goal is to find a good estimate for the number of real roots of $P$ in $(a,b)$ (counted with multiplicity).

We start with an elementary observation known as \emph{Descartes's law of signs}: the number of sign variations in the coefficients\footnote{When counting the sign variations, zeroes are ignored, \ie the number of sign variations in a list is equal to the number of sign variations of the list with all zeroes removed (with the convention that the number of sign variations in the empty list is zero).} of $P$ is an upper bound for the number of \emph{positive} real roots, counted with multiplicity, and the difference between the number of sign variations and the number of positive real roots is even. In particular, if there are no sign variations then $P$ has no positive real roots, and if there is exactly one sign variation then $P$ has a unique (and simple) positive real root.
In order to use this idea to estimate the number of roots in $(a,b)$ we apply a transformation to $P$ to obtain a polynomial $P[a,b]$ whose positive real roots correspond to the roots of $P$ in $[a,b]$.
We define three basic transformations: 
\begin{itemize}
\item $\T_c(P(x)) = P(x - c)$.
\item $\Co_\lambda(P(x)) = P(\lambda x)$. 
\item $\Rec_d(P(x)) = x^dP(1/x)$. 
\end{itemize}
Now consider the polynomial
\[P[a,b] = \T_{-1} \circ \Rec_{d} \circ \Co_{b - a} \circ \T_{-a} (P). \]
Intuitively, the roots in the open interval $(a,b)$ are first shifted to the interval $(0, b - a)$, then contracted into the interval $(0,1)$, then their reciprocal is taken, sending them to $(1, \infty)$, and finally they are shifted to $(0, \infty)$. Thus, the roots of $P$ in $(a,b)$ correspond to the roots of $P[a,b]$ in $(0, \infty)$.
Interestingly, the application of this transformation can be viewed as a change of basis: The Bernstein polynomials of degree $d$ for $a, b$ are
\[B_{d,i}(a,b) = {d \choose i}\frac{(x - a)^i(b - x)^{d - i}}{(b - a)^d}. \]
These polynomials form a partition of unity and a basis of polynomials of degree at most $d$. One can show that $P[a,b]$ is just the representation of $P$ in the Bernstein basis for $a,b$ of degree $d$. 
Hence, we can estimate the number of real roots of $P$ on $a,b$ by first translating $P$ into the Bernstein basis and then counting the sign variations. Let us proceed to show that this estimate is sufficiently good to yield a polynomial-time root counting algorithm.

If $L$ is a list of numbers, we denote by $\var(L)$ the number of sign variations in this list.
Let us denote the coefficients of $P$ in the Bernstein basis of degree $\deg P$ for $a, b$ by $b(P,a,b)$. Let $\IQ$ denote the space of compact intervals with rational endpoints, including degenerate intervals (\ie points). 
The following Lemma combines Descartes's law of signs with the ``Theorem of three circles'' from \cite{BasuPollackRoy}.

\begin{Lemma}\label{Lemma: sign variations and root count}
Let $P \in \Z[x]$, let $[a,b] \in \IQ$ be a compact rational interval. Then
\begin{enumerate}
\item If $\var(b(P,a,b)) = 0$ then $P$ has no roots in $(a,b)$.
\item If $\var(b(P,a,b)) = 1$ then $P$ has a unique root in $(a,b)$.
\item If $P$ has no complex roots in the disk with diameter $[a,b]$, then $\var(b(P,a,b)) = 0$.
\item If $P$ has a unique simple complex root in the union of the two disks which circumscribe the equilateral triangles based in $[a,b]$, then $\var(b(P,a,b)) = 1$.
\end{enumerate}
\end{Lemma}

The coefficients of a polynomial in the Bernstein basis can be computed in polynomial time. If the coefficients for an interval $[a,b]$ are known, the coefficients for subintervals $[a,m]$ and $[m,b]$ can be computed from these coefficients in a more efficient manner:

\begin{Lemma}[\cite{BasuPollackRoy}]\label{Lemma: Bernstein coefficients}\hfill
\begin{enumerate}
\item There exists a polytime algorithm which takes as input a polynomial $P \in \Z[x]$ and an interval $[a,b] \in \IQ$ and outputs the list of coefficients of $P[a,b]$.
\item There exists a polytime algorithm which takes as input an interval $[a,b] \in \IQ$, a point $m \in \Q$ (not necessarily in $[a,b]$), and the list of coefficients $b(P,a,b)$ of a polynomial $P$ represented in the Bernstein basis for $a, b$, and outputs the coefficients of $P$ in the Bernstein bases $b(P,a,m)$ for $a,m$ and $b(P,m,b)$ for $m,b$ respectively. 
\end{enumerate}
\end{Lemma}

The root counting technique we have sketched here can be used to isolate all real roots of a polynomial $P$ in an interval $[a,b]$ in polynomial time. 
In order to get rid of multiple roots, we first compute the separable part $\tilde{P}$ of $P$, which can be done in polynomial time using signed subresultant sequences (see \eg \cite[Algorithm 8.23]{BasuPollackRoy}). Then we translate to $\tilde{P}[a,b]$, which by Lemma \ref{Lemma: Bernstein coefficients}.1 can be done in polynomial time, and count the real roots. If the result is different from $0$ or $1$, we use Lemma \ref{Lemma: Bernstein coefficients}.2 to compute $\tilde{P}[a,m]$ and $\tilde{P}[m,b]$ and apply this idea recursively, removing all intervals with zero sign variations, keeping all intervals with one sign variation, and splitting all intervals with more sign variations. For a proof that this will take polynomial time see \cite[Algorithm 10.5]{BasuPollackRoy}.

This in turn yields an algorithm for computing the global maximum of a polynomial $P$ on an interval $[a,b]$ in polynomial time: isolate the real roots of the separable part of the derivative $P'$. 
Use the bisection method to approximate the roots up to sufficient accuracy. Evaluate $P$ on the approximate roots and on the endpoints of the interval and take the maximum over the list of results. 

Our maximisation algorithm can be viewed as a combination of this idea with the subdivision scheme used in Algorithm \ref{Algorithm: maximisation by subdivision}.

\subsection{A generic interface for range computation}

Our algorithm provides a generic interface that allows us to use it in several different contexts. For an interval $I$, let 
$|I| = \operatorname{I}/2$ denote its \emph{radius}.

\begin{Definition}\label{Definition : local maximisation data}
Let $f \colon [l,r] \to \R$ be a real function defined on some compact interval. 
A tuple
\[
  (a, b, n, F, B, G) \in \maxdatatype
\]
is called \emph{local maximisation data} for $f$ on $[a,b]$ with accuracy $n$ if:
\begin{enumerate}
\item $[a,b] \subseteq [l,r]$
\item $f(I) \subseteq F(I)$ for all intervals $I \subseteq [a,b]$.
\item $|F(I)| \leq L \cdot |I| + 2^{-n}$ for some $L \in \R$ and all $I \in \IQ$.
\item $F(I) \to F(J)$ whenever $I \to J$ in the Hausdorff metric.
\item $B = b(P,a,b)$ for some polynomial $P$ which has the same roots in $[a,b]$ as the derivative of some function $h\colon [a,b] \to \R$ with $h(I) \subseteq F(I)$ for all $x \in [a,b]$.
\item $G(x)$ and $P(x)$ have the same sign for all $x \in \Q\cap [a,b]$.
\end{enumerate}
Note that the third condition in particular applies to degenerate intervals, and hence implies that $|F(x)| \leq 2^{-n}$ for all $x \in [a,b]$.
\end{Definition}

\begin{Definition}\label{Definition: maximiser}
Let $f \colon [l,r] \to \R$ be a real function defined on some compact interval. A \emph{maximiser} $\mathcal{M}$ for $f$ is a function
\[
\mathcal{M} \colon \D \times \D \times \N 
\to 
\left(\maxdatatype\right), 
\]
such that $\mathcal{M}(a, b, n)$ is local maximisation data for $f$ on $[a,b]$ with accuracy $n$, subject to the following two \emph{monotonicity conditions}:
\begin{enumerate}
\item If $[a',b'] \subseteq [a,b]$ and $n' \leq n$ then the size of the last three components of
$\mathcal{M}(a',b', n')$ is smaller than the size of the last three components of $\mathcal{M}(a, b, n)$.
\item If $[a',b'] \subseteq [a,b]$ and $n' \geq n$ then, if $\mathcal{M}(a',b', \varepsilon') = (a',b', n', F',B',G')$ and 
$\mathcal{M}(a,b, n) = (a,b,n,F,B,G)$, then $F'([c,d]) \subseteq F([c,d])$ for all 
$[c,d] \subseteq [a',b']$.
\end{enumerate}
Note that the inequality for $n$ and $n'$ in the second condition is the opposite of the inequality in the first condition.
\end{Definition}

The generic interface is used by the representations $\Poly$, $\PPoly$, and $\Frac$ as well as by their local counterparts.
Recall that a $\Local \Poly$-name of a function $f \colon [l,r] \to \R$ is a function 
\[
  A \colon \D \times \D \times \N \to \D[x]
\]
satisfying
\[
  |A(a,b,n)(x) - f(x)| < 2^{-n}
\]
for all $x \in [l,r]$.
A $\Poly$-name can be viewed as a special case of this, where $A(a,b,n)$ is independent of $a$ and $b$.

Given $A$ we can compute a maximiser $\mathcal{M}$ for $f$ as follows:
Given $(a,b,n) \in \D \times \D \times \N$,
compute the polynomial $P = A(a,b,n)$
and translate it to the monomial
basis.
Let
\[
  \mathcal{M}(a,b,n) = (a, b, n, F, B, G)
\]
where $F(x) = P(x)$ is the evaluation function of the polynomial $P$,
the vector
$B = b(\alpha \cdot P, a, b)$
represents $\alpha P$ in the Bernstein basis on $[a,b]$, 
where $\alpha$ is the $\operatorname{lcm}$ of the denominators of the coefficients of $P$ in the monomial basis,
and $G(x) = P'(x)$.
Here, the Bernstein coefficients $B$ are computed from a representation in the monomial basis as outlined in the beginning of Section \ref{Section: Root counting}.

Analogously, a $\Local \Frac$-name of a function $f \colon [l,r] \to \R$ is a function 
\[
  A \colon \D \times \D \times \N \to \D(x)
\]
satisfying
\[
  |A(a,b,n)(x) - f(x)| < 2^{-n}
\]
for all $x \in [l,r]$.
Again, a $\Frac$-name can be viewed as a special case of this, where $A(a,b,n)$ is independent of $a$ and $b$.

Given $A$ we can compute a maximiser $\mathcal{M}$ for $f$ as follows:
Given $(a,b,n) \in \D \times \D \times \N$,
compute the rational function $P/Q = A(a,b,n)$
and translate both $P$ and $Q$ into the monomial basis.
Let
\[
  \mathcal{M}(a,b,n) = (a, b, n, F, B, G)
\]
where 
$F(x) = P(x)/Q(x)$ is the evaluation function of the rational function $P/Q$, ${B = b(\alpha (PQ' - QP'), a,b)}$, 
$\alpha$ being the $\operatorname{lcm}$ of the denominators of the coefficients of $PQ' - QP'$ in the monomial basis,
and
${G(x) = (P/Q)'(x)}$.

\subsection{The maximisation algorithm for approximation-based representations}

\begin{Definition}
Let $f \colon [l,r] \to \R$ be a continuous real function.
 Let $L = (a,b,n,F,B,G)$ be local maximisation data for $f$, where $B = b(P,a,b)$.
A \emph{maximisation interval} for $f$ with data $L$ is given by a union type with two variants:
\begin{enumerate}
\item A \emph{search interval} is a tuple
\[(c, d, n, v, G, C) \in \D \times \D \times \N \times \ID \times (\Q \to \Q) \times \Z^*,\]
where $[c,d] \subseteq [a,b]$, $v$ contains the maximum of $f$ over the interval $[c,d]$, and
$C = b(\alpha P, c, d)$ for some constant $\alpha$.
\item A \emph{critical interval} is a tuple
\[(c, d, n, v) \in  \D \times \D \times \N \times \ID, \]
where $[c,d] \subseteq [a,b]$, and $v$ contains the maximum of $f$ over the interval $[c,d]$
\end{enumerate}
In both cases we call $c$ and $d$ the \emph{endpoints} of the maximisation interval, $v$ the \emph{value}, and $n$ the \emph{accuracy}.
\end{Definition}
Maximisation intervals are endowed with the following total preorder: 
Let $M_1$ and $M_2$ be maximisation intervals for $f$, not necessarily associated with the same local approximation data. Then $M_1 \leq M_2$ if and only if the right endpoint of the value of $M_1$ is smaller than the right endpoint of the value of $M_2$.
Our algorithm relies on two auxiliary algorithms for creating search intervals:

\begin{algorithm}[H]
\caption{Creating a maximisation interval}\label{Algorithm: creating maximisation interval}
\begin{algorithmic}
\REQUIRE Local maximisation data $L = (a,b,n,F,B,G)$ for a function $f$.
\ENSURE A maximisation interval for $f$ with data $L$ whose endpoints are $a$ and $b$.
\STATE \textbf{Procedure:}
\STATE - Count the number $v$ of sign variations in $B$.
\IF {$v = 0$}
\RETURN  The critical interval $\left(a,b, n, \max\left(F(a), F(b)\right)\right)$.
\ELSIF {$v = 1$}
\STATE  - Use binary search on $G$ to determine a small interval $[a',b'] \subseteq [a,b]$ which contains the unique 
	   zero of the polynomial with Bernstein coefficients $B$ in $[a,b]$, such that $F([a',b'])$ has radius at most $2^{-n + 1}$.
\RETURN The critical interval $\left(a,b, n, F([a',b'])\right)$
\ELSE
\RETURN The search interval $\left(a,b, n, F([a,b]), G, B\right)$.
\ENDIF
\end{algorithmic}
\end{algorithm}

\begin{algorithm}[H]
\caption{Splitting a search interval}\label{Algorithm: splitting search interval}
\begin{algorithmic}
\REQUIRE A search interval $(a,b,n,v,F,G,B)$ for $f$ with local data $L$. A number $m \in (a,b)$.
\ENSURE Two maximisation intervals $M_l$ and $M_r$ for $f$ with local data $L$ with accuracy $n$ and endpoints $(a,m)$ and $(m,b)$ respectively.
\STATE \textbf{Procedure:}
\STATE - Compute the coefficients of the polynomial represented by $B$ in the Bernstein basis for $[a,m]$ and $[m,b]$ from $B$ using Lemma \ref{Lemma: Bernstein coefficients}.
\STATE - Count the number of sign variations in the left and right coefficients. 
\STATE - Based on the number of sign variations, proceed as in Algorithm \ref{Algorithm: creating maximisation interval}.
\end{algorithmic}
\end{algorithm}

We are now ready to describe our maximisation algorithm.

\begin{algorithm}[H]
\caption{Maximisation}\label{Algorithm: local maximum}
\begin{algorithmic}
\REQUIRE A {maximiser} $\mathcal{M}$ for a function $f \colon [l,r] \to \R$.
An accuracy requirement $n \in \N$.
\ENSURE An interval $I$ of radius smaller than $2^{-n}$ which contains the global maximum of $f$ over the interval $[l,r]$.
\STATE \textbf{Procedure:}
\STATE - Query the maximiser $\mathcal{M}$ for an initial local approximation $L_0 = (l,r,n_0,F,B,G)$ on $[l,r]$ with accuracy $n_0 = 1$.
\STATE - Apply Algorithm \ref{Algorithm: creating maximisation interval} to obtain a maximisation interval $M_0$ for $f$ with data $L_0$.
\STATE - Create a priority queue $\mathbf{Q} = \{M_0\}$.
\LOOP
	\STATE - Remove the largest element $M$ from the queue $\mathbf{Q}$.
	\IF {The value of $M$ has radius smaller than $2^{-n}$}
		\RETURN The value of $M$.
	\ELSE
		\IF {$M$ is a critical interval}
		\STATE - Let $n_M$ denote its accuracy and let $a$ and $b$ denote its endpoints.
    \STATE - Query $\mathcal{M}$ for local approximation data for $f$ on $[a,b]$ with accuracy $2^{-n - 1}$.
		 \STATE - Use Algorithm \ref{Algorithm: creating maximisation interval} to compute a corresponding maximisation interval $M'$.
		 \STATE - Add $M'$ to the priority queue $\mathbf{Q}$.
		\ELSE  
		\STATE If $M$ is a search interval, let $v$ denote its value and let $a$ and $b$ denote its endpoints.
		\IF {The radius of $v$ is smaller than $2^{-n_M + 1}$, where $n_M$ is the accuracy of $M$}
			\STATE - Query $\mathcal{M}$ for local maximisation data $L$ for $f$ on $[a,b]$ with accuracy $n_M + 1$.
			\STATE - Use Algorithm \ref{Algorithm: creating maximisation interval} to compute a maximisation interval $M'$ for $f$ with 						 data $L$.
		\STATE - Add $M'$ to the priority queue $\mathbf{Q}$.
		\ELSE
			\STATE - Let $m = (a + b)/2$.
			\STATE - Use Algorithm \ref{Algorithm: splitting search interval} with input $M$ and $m$ to create two new maximisation							 intervals $M_l$ and $M_r$.
			\STATE - Add $M_l$ and $M_r$ to the the priority queue $\mathbf{Q}$.
		\ENDIF
		\ENDIF
	\ENDIF
\ENDLOOP
\end{algorithmic}
\end{algorithm}

\begin{Theorem}
Algorithm \ref{Algorithm: local maximum} is correct.
\end{Theorem}
\begin{proof}
The queue will never be empty, for whenever an interval is removed from the priority queue, either the algorithm terminates or at least one new interval is added to the queue. The algorithm will terminate eventually, for if a maximisation interval $M$ is removed from the queue, then by the Lipschitz condition on the function $F$ (Definition \ref{Definition : local maximisation data}.3) and the monotonicity of the maximiser (Definition \ref{Definition: maximiser}.2) the radius of the values of the maximisation intervals which are added to the queue is at most half the radius of the value of $M$, and the algorithm terminates as soon as an interval with sufficiently small radius is removed.

Let $m = \max_{x \in [-1,1]} f(x)$. By construction, in every iteration of the loop, one of the maximisation intervals in the queue contains $m$. We claim that if an interval $M$ with value $v = [v_l,v_r]$ is removed from the queue then $v$ has to contain $m$.
Since $[v_l,v_r]$ contains a value of $f$, we have $v_l \leq m$. Thus, if $[v_l,v_r]$ does not contain the maximum, then $m > r$. 
But there exists some interval $M$ in the queue which contains the maximum of $f$. Let $v' = [v_l',v_r']$ denote its value.
Then we have $m \in [v_l', v_r']$ and hence $v_r' \geq m > v_r$, contradicting the fact that $M$ is removed first.
It follows that if $M$ is removed and the radius of its value $v$ is smaller than $2^{-n}$, then $v$ is a valid output.
\end{proof}

In order to estimate the running time of Algorithm \ref{Algorithm: local maximum}, we introduce an auxiliary algorithm which is easier to analyse and whose running time dominates the running time of Algorithm \ref{Algorithm: local maximum}.

\begin{algorithm}[H]
\caption{Slow polynomial maximisation}\label{Algorithm: slow maximisation}
\begin{algorithmic}
\REQUIRE A {maximiser} $\mathcal{M}$ for a function $f \colon [l,r] \to \R$. An accuracy requirement $n \in \N$.
\ENSURE An interval $I$ of radius smaller than $2^{-n}$ which contains the global maximum of $f$ over the interval $[l,r]$.
\STATE \textbf{Procedure:}
\STATE - Query $\mathcal{M}$ for a local approximation $L = (l,r,n + 1,F,B,G)$ on $[l,r]$ with accuracy $n + 1$.
\STATE - Apply Algorithm \ref{Algorithm: creating maximisation interval} to obtain a maximisation interval $M_0$ for $f$ with data $L$.
\STATE - Create a priority queue $\mathbf{Q} = \{M_0\}$.
\LOOP
	\STATE - Remove the largest element $M$ from the queue $\mathbf{Q}$.
	\IF {The value of $M$ has radius smaller than $2^{-n}$}
		\RETURN The value of $M$.
	\ENDIF
	\FORALL {$M \in \mathbf{Q}$}
		\IF {If $M$ is a search interval and the radius of its value is bigger than $2^{-n}$}
			\STATE - Let $m = (a + b)/2$.
			\STATE - Use Algorithm \ref{Algorithm: splitting search interval} with input $M$ and $m$ to create two new maximisation							 intervals $M_l$ and $M_r$.
			\STATE - Add $M_l$ and $M_r$ to the the priority queue $\mathbf{Q}$.
		\ENDIF
	\ENDFOR
\ENDLOOP
\end{algorithmic}
\end{algorithm}

In contrast to Algorithm \ref{Algorithm: local maximum}, Algorithm \ref{Algorithm: slow maximisation} uses a single global approximation rather than local approximations. It processes every element of the priority queue in every step, rather than just looking at the ``most promising'' element.
It thus simulates a certain worst-case scenario for Algorithm \ref{Algorithm: local maximum}.
By the monotonicity assumption on maximisers, the running time of Algorithm \ref{Algorithm: local maximum} is majorised by the running time of Algorithm \ref{Algorithm: slow maximisation}.
In order to estimate the running time of Algorithm \ref{Algorithm: slow maximisation} we need a more technical estimate on the growth of the bitsize of the coefficients in the Bernstein basis.

\begin{Lemma}[\cite{BasuPollackRoy}]\label{Lemma: Bernstein coefficients size bound}
Consider the algorithm from Lemma \ref{Lemma: Bernstein coefficients}.2 which computes from the Bernstein coefficients $b(P,l,r)$, the coefficients $b(P,l,m)$ and $b(P,m,r)$. If $\beta$ is a bound on the bitsize of the elements of $b(P,l,r)$ and $\beta'$ is a bound on the bitsize of $l$ and $r$, then the bitsize of the elements of $b(P,l,m)$ and $b(P,m,r)$ is bounded by $\operatorname{(\deg P + 1)\beta' + \beta}$.
\end{Lemma}

\begin{Theorem}\label{Theorem: maximisation in polynomial time}
Algorithm \ref{Algorithm: slow maximisation} runs in polynomial time.
\end{Theorem}
\begin{proof}[Proof (Sketch)]
Let $L = (l,r,n + 1,F,B,G)$ be the local maximisation data which is used by the algorithm, where $B = b(P,l,r)$. The main thing to show is that the number of intervals the algorithm processes is bounded polynomially in $n$ and the size of $B$. The intervals and coefficients considered in the algorithm can be arranged in a binary tree as follows:

\begin{itemize}
\item Label the root with the interval $([l,r], B)$.
\item If $([a,b], b(P,a,b))$ is the label of a node of the tree...
\begin{itemize}
\item If $F([a,b])$ has diameter smaller than $2^{-n}$, then the node is a leaf. Otherwise...
\item If $\var(b(P,a,b)) = 0$ the node is a leaf.
\item If $\var(b(P,a,b)) = 1$ the node is a leaf.
\item If $\var(b(P,a,b)) > 1$ the node has two successors, labelled respectively with $([a,m], b(P,a,m))$ and $([m,b], b(P,m,b))$.
\end{itemize}
\end{itemize}

It follows from the Lipschitz condition on $F$ that the height of the tree is bounded by $O(n)$. 
At each level of the tree, consider the number of nodes $([a,b],b(P,a,b))$ with $\var(b(P,a,b)) > 0$. By Lemma \ref{Lemma: sign variations and root count}.3, each such node can be associated with a complex root of the polynomial $P$. 
Hence there are at most $\deg P$ nodes which aren't leaves at each level. It follows that the number of nodes in the tree is bounded polynomially in $n$ and $\deg P$.
It remains to show that the size of the Bernstein coefficients associated with each interval is bounded polynomially in $n$ and the bitsize of $B$. Let us write $t_P = \deg P + 1$ for the number of terms in $P$. Let $\beta$ denote a bound on the bitsize of the elements of $B = b(P,l,r)$ and $\beta'$ denote a bound on the bitsize of $l$ and $r$. Then by Lemma \ref{Lemma: Bernstein coefficients size bound}, the size of $b(P,a,m)$ and $b(P,m,r)$ is bounded by $\operatorname{t_P\beta' + \beta}$ and the size of $m$ is bounded by $\beta' + 2$. It follows that the size of the Bernstein coefficients on the $n^{\text{th}}$ level of the tree is bounded by 
\[nt_P\left(\beta' + 2(n - 1) \right) + \beta. \]
Hence the bitsize is bounded by $nt_P\left(\beta' + 2(n - 1) \right) + \beta$, which is polynomial in $n$ and the bitsize of $B$. 
Hence we perform polynomially many operations on objects of polynomially bounded size, so that the algorithm runs in polynomial time. 
\end{proof}

\section{Root finding for polynomials}

The real root counting technique we have just described is also used to find the roots of a real polynomial.
This is an important subroutine in the pointwise maximisation algorithm for polynomials and piecewise polynomials.
Given a dyadic polynomial $P$ we can in principle compute complete information on the set of all roots of $P$ in the sense that 
we can compute for each accuracy requirement $n \in \N$ a list of intervals $(I_m)_m$ such that each root of $P$ is contained in one the intervals $I_m$, each interval contains a root of $P$, and each interval has radius at most $2^{-n}$.

However, our root counting technique cannot directly be used to achieve this when the polynomial has multiple roots, as roots are counted with multiplicity and the existence of a root is only guaranteed when the number of sign variations is equal to $1$.
One could in principle eliminate all multiple roots by computing the separable part first (see \cite[Algorithm 8.23]{BasuPollackRoy}),
but for our purpose it suffices to compute an upper bound to the set of roots.

This functionality is provided by the Haskell function \lstinline{findRootsWithEvaluation}
which is found in the module
\lstinline{AERN2.Poly.Power.RootsIntVector}.
In addition to an integer polynomial $P$ it takes as its second input a function $e$ which computes a value of interest on each rational interval
and a function $\operatorname{OK?}$ which checks if the value of interest meets a certain requirement. 
It then outputs a list of intervals such that each root of the polynomial $P$ is contained in one of the intervals together with the values of interest on those intervals.

\begin{algorithm}[H]\label{Algorithm: root finding}
  \caption{Polynomial Root Finding with Evaluation}
  \begin{algorithmic}
  \REQUIRE 
    An interval $[l,r]$ with rational endpoints.
    An integer polynomial $P$ in the monomial basis.
    A function $e \colon \IQ \to A$ where $A$ is a discrete set.
    A function $\operatorname{OK?} \colon A \to \{0,1\}$.
  \ENSURE A finite list $(I_m, v_m)$ of pairs of a rational interval $I_m$ and the value $v_m = e(I_m)$
  such that $\operatorname{OK?}(v_m) = 1$ for all $m$ and every root of $P$ is contained in one of the intervals $I_m$.
  \STATE \textbf{Procedure:}
  \STATE - Compute the Bernstein coefficients $b(P,l,r)$ of $P$.
  \STATE - Initialise a list $\mathbf{L} = \left\{([l,r], b(P,l,r))\right\}$.
  \STATE - Initialise a list $\mathbf{Res} = \{\}$.
  \LOOP 
  \FORALL {$(I, b)$ in the list $\mathbf{L}$}
  \STATE - Compute the sign variations $\var(b)$ in $b$.
  \IF {The number of sign variations is $1$}
    \STATE Bisect the interval $I$ into two intervals $I_1$ and $I_2$.
    \STATE Compute the unique index $k \in \{1,2\}$ such that $P$ changes its sign on the endpoins of $I_k$.
    \STATE Compute the value $v = e(I_k)$.
    \IF {$\operatorname{OK?}(v) = 1$}
    \STATE Add $(I_k, v)$ to the list $\mathbf{Res}$.
    \ELSE  
    \STATE Compute the Bernstein coefficients $b(P, I_k)$.
    \STATE Add $(I_k, b(P, I_k))$ to the list $\mathbf{L}$.
    \ENDIF 
  \ELSIF {The number of sign variations is greater than $1$}
    \STATE Bisect the interval $I$ into two intervals $I_1$ and $I_2$.
    \FORALL {$k \in \{1,2\}$}
    \STATE Compute the value $v_k = e(I_k)$.
    \IF {$\operatorname{OK?}(v_k) = 1$}
    \STATE Add $(I_k, v_k)$ to the list $\mathbf{Res}$
    \ELSE 
    \STATE Compute the Bernstein coefficients $b(P, I_k)$.
    \STATE Add $(I_k, b(P, I_k))$ to the list $\mathbf{L}$.
    \ENDIF
    \ENDFOR
  \ENDIF
  \ENDFOR
  \ENDLOOP
  \RETURN The list $\mathbf{Res}$
  \end{algorithmic}
  \end{algorithm}

  It follows from Lemma \ref{Lemma: sign variations and root count} that Algorithm \ref{Algorithm: root finding} is correct whenever it terminates.
  It is guaranteed to terminate if for every $x \in [l, r]$ there exists an $\varepsilon > 0$ such that $\operatorname{OK?}(e(I)) = 1$ for all intervals $I$ containing $x$ with $|I| < \varepsilon$.

\section{Pointwise maximisation}

\subsection{Pointwise maximisation for local and evaluation-based representations}

Computing the pointwise maximum of two functions with respect to the $\BFun$ representation is straightforward: 
simply compute the pointwise maximum using interval arithmetic.

To compute the pointwise maximum with respect to $\DBFun$, we need to take into account that this potentially introduces points of non-differentiability.
Thus, given $\DBFun$-names $(\varphi, \varphi')$ and $(\psi, \psi')$ of functions $f$ and $g$, a $\DBFun$-name $(\xi, \xi')$ of the pointwise maximum is computed as follows:
$\xi$ is the pointwise maximum of the $\BFun$-names $\varphi$ and $\psi$.
To compute $\xi'(I)$, first compare $\phi(I)$ and $\psi(I)$.
If $\varphi(I)$ is certainly greater than $\psi(I)$ then let $\xi'(I) = \varphi'(I)$.
If $\psi(I)$ is certainly greater than $\varphi(I)$ then let $\xi'(I) = \psi'(I)$. 
If neither is the case, let $\xi'(I)$ be the interval hull of $\varphi'(I)$ and $\psi'(I)$, \ie the smallest interval containing both $\varphi'(I)$ and $\psi'(I)$.

Given a $\Fun$-name $(\varphi, \omega)$ of a function $f$ and a $\Fun$-name 
$(\psi, \mu)$ of a function $g$, 
a $\Fun$-name of the pointwise maximum $\max\{f, g\}$ 
is given by 
$(\xi, \nu)$,
where 
$\xi = \max\{\varphi, \psi\}$
and 
$\nu(n) = \max\{\omega(n + 1), \mu(n + 1)\}$.

The pointwise maximum for ``local'' representations, like any other binary operation, is computed by lifting the corresponding operation for the ``global'' representation. 

\subsection{Pointwise maximisation for piecewise polynomials}

Pointwise maximisation for piecewise polynomials is provided by the module \lstinline{AERN2.PPoly.MinMax}.
The computation of the pointwise maximum of two piecewise polynomials is easily reduced to the computation of the pointwise maximum of two polynomials as a piecewise polynomial.
This is achieved by the following algorithm:

\begin{algorithm}[H]
  \caption{Piecewise polynomial maximisation}\label{Algorithm: PPoly pointwise maximum}
  \begin{algorithmic}
  \REQUIRE A pair of dyadic polynomials $P$ and $Q$ on a domain $D$. An accuracy requirement $n \in \N$.
  \ENSURE A piecewise polynomial $f$ satisfying 
    $|f(x) - \max\{P(x),Q(x)\}| \leq 2^{-n}$
    for all
    $x \in D$.
  \STATE \textbf{Procedure:}
  \STATE - Initialise an empty list $\mathbf{L} = \{\}$.
  \STATE - Let $C = P - Q$.
  \STATE - Use Algorithm \ref{Algorithm: root finding} to compute a finite list of intervals $(I_m)_m$ such that each root of $C$ is contained in one of the intervals and $|C(x)| < 2^{-n}$ for all $m$ and all $x \in I_m$.
  More formally, call Algorithm \ref{Algorithm: root finding} with the input polynomial being $C$, the evaluation function $e$ being evaluation of $C$ over an interval using ball arithmetic, and the function $\operatorname{OK?}$ being the function that checks if a given interval has radius smaller than $2^{-n}$.
  \FORALL {$m$}
  \STATE - Add the pair $(I_m, P)$ to the list $\mathbf{L}$.
  \ENDFOR
  \STATE - Extend the list $(I_m)_m$ to a partition of the domain $D$,
  \ie compute a list of disjoint intervals $(J_k)_k$ such that the $J_k$'s are disjoint from the $I_m$'s except at the endpoints and $D$ is covered by the union of the $J_k$'s and $I_m$'s.
  \FORALL {$k$}
    \STATE - Let $c$ be the centre of $J_k$.
    \IF {$P(c) >= Q(c)$}
    \STATE - Add the pair $(J_k, P)$ to the list $\mathbf{L}$. 
    \ELSE 
    \STATE - Add the pair $(J_k, Q)$ to the list $\mathbf{L}$.
    \ENDIF
  \ENDFOR
  \RETURN The piecewise polynomial encoded by the list $\mathbf{L}$.
  \end{algorithmic}
  \end{algorithm}

  The correctness of Algorithm \ref{Algorithm: PPoly pointwise maximum} follows immediately from that of Algorithm \ref{Algorithm: root finding}.
  The algorithm terminates as the diameter of $C(I)$ converges to zero as the diameter of $I$ converges to zero.
  It can in fact be shown to run in polynomial time by using similar ideas as for Theorem \ref{Theorem: maximisation in polynomial time}.

\subsection{Pointwise maximisation for polynomials}

Pointwise maximisation of polynomials is provided by the module 
\lstinline{AERN2.Poly.Cheb.MinMax}.

To compute the pointwise maximum of two polynomials $P$ and $Q$, we first use the Algorithm \ref{Algorithm: local maximum} to attempt to prove that $P(x) >= Q(x)$ or $Q(x) >= P(x)$ for all $x$ in the domain.
If this succeeds we output $P$ or $Q$ respectively.

If it fails, we first compute a Chebyshev interpolation $h$ of the function $\max\{P,Q\}$ using an algorithm described in \cite{BT97} based on the Discrete Cosine Transform.

We then bound the approximation error $|h - \max\{P,Q\}|$ essentially by computing the pointwise maximum of $P$ and $Q$ as a piecewise polynomial $f$ as in Algorithm \ref{Algorithm: PPoly pointwise maximum} and by computing the maximum of $|h(x) - f(x)|$ on each piece, using Algorithm \ref{Algorithm: local maximum}.


\section{Integration}

\subsection{Integration for evaluation-based representations}

Integration of functions using the evaluation-based representation is performed using Riemann sums:

\begin{algorithm}[H]
  \caption{Evaluation-based integration}\label{Algorithm: Fun Integration}
  \begin{algorithmic}
  \REQUIRE 
    An interval inclusion $F \colon \ID \to \ID$ of a continuous real function $f \colon D \to \R$.
    A sequence $(p_i)_i \subseteq \N$ of precisions.
    Two dyadic numbers $l, r \in D$ with $l < r$.
    An accuracy requirement $n \in \N$.
  \ENSURE 
    An interval $I$ of radius at most $2^{-n}$ containing the value of the integral $\int_l^r f(x) \operatorname{dx}$.
  \STATE \textbf{Procedure:}
  \STATE - Compute the value $F([l,r]) \cdot |l - r|$ using precision $p_0$ and write the result in a variable $A_0$.
  \IF {The radius of $A_0$ is smaller than $2^{-n}$}
  \RETURN $A_0$.
  \ELSE
    \STATE - Compute the value $F([l,r]) \cdot |l - r|$ using precision $p_1$ and write the result in a variable $A_1$.
    \STATE - Let $k = 0$.
    \WHILE {The radius of $A_1$ is strictly smaller than that of $A_0$}
    \STATE - Put $A_0 = A_1$.
    \STATE - Compute the value $F([l,r]) \cdot |l - r|$ using precision $p_{k + 2}$ and write the result in the variable $A_1$.
    \IF {The radius of $A_1$ is smaller than $2^{-n}$}
    \RETURN $A_1$.
    \ENDIF
    \STATE - Put $k = k + 1$.
    \ENDWHILE
    \STATE - Let $m = (l + r)/2$.
    \STATE - Call Algorithm \ref{Algorithm: Fun Integration} with function inclusion $F$, precision sequence $(p_{i + k})_i$, endpoints $l$ and $m$, and accuracy requirement $n + 1$.
    Call the result $I_0$.
    \STATE - Call Algorithm \ref{Algorithm: Fun Integration} with function inclusion $F$, precision sequence $(p_{i + k})_i$, endpoints $m$ and $r$, and accuracy requirement $n + 1$.
    Call the result $I_1$.
    \RETURN $I_0$ + $I_1$.
  \ENDIF 
  \end{algorithmic}
  \end{algorithm}

  The termination of Algorithm \ref{Algorithm: Fun Integration} follows from the assumption that the interval inclusion $F$ converges to $f$ as the precision is increased and the fact that the Riemann sums for $f$ converge to the integral of $f$ as the mesh width of the subdivision converges to zero, together with the continuity of the integral operator on the space $C(D)$.
  The algorithm is easily seen to be correct.

 \subsection{Integration for representations based on global approximations}

 Integration of a function $f$ given by a sequence of polynomial approximations $(P_n)_n$ is performed in the obvious way: Find a sufficiently accurate polynomial approximation $P_n$ and integrate it symbolically.
 When $f$ is given by a sequence of piecewise polynomial approximations, an analogous method is used.

 To integrate a function given by a sequence of rational approximations, 
 we first find a sufficiently good rational approximation, then translate it into a piecewise polynomial, using the division algorithm described in \cite{MainArticle}, which is then integrated.

 \subsection{Integration for locally approximating representations}

  The integration of functions given by a Haskell type of the form $\Local a$ is reduced to the integration for the type $a$ by a simple heuristic:
  Given an accuracy requirement $n \in \N$, subdivide the domain into $n$ pieces of the same size and compute the integral over each piece to accuracy $n + \lceil\log_2 (n)\rceil + 1$,
  and add up the results.

\bibliographystyle{abbrv}
\bibliography{fnreps-technical}
\nocite{*}

\end{document}